\documentclass[journal,twoside,web]{ieeecolor}
\usepackage{generic}
\usepackage{cite}
\usepackage{amsmath,amssymb,amsfonts,bbm}
\usepackage{algorithmic}
\usepackage{graphicx}
\usepackage{algorithm,algorithmic}
\usepackage{soul} 
\usepackage{xcolor}
\usepackage{tikz} 
\usepackage{color}
\usepackage{graphicx}
\usepackage{textcomp}
\usepackage{trimclip}
\usepackage[mathscr]{euscript}
\usepackage{array}
\usepackage{eqparbox}
\usepackage{url}
\usepackage{relsize}
\usepackage{dsfont}
\usepackage[mathscr]{euscript}
\usepackage[colorlinks=true, allcolors=blue]{hyperref}

\newcommand{\calX}{{\mathcal X}}
\newcommand{\calY}{{\mathcal Y}}

\newcommand{\bc}{{\mathbf c}}
\newcommand{\bx}{{\mathbf x}}
\newcommand{\by}{{\mathbf y}}

\newcommand{\bk}{{\mathbf k}}

\newcommand{\vsp}{\vspace{0.017in}}

\newtheorem{theorem}{Theorem}

\newtheorem{proposition}[theorem]{Proposition}
\newtheorem{problem}{Problem}
\newtheorem{definition}{Definition}

\newtheorem{remark}{Remark}

\newcommand{\blue}{\color{blue}}

\begin{document}

\bstctlcite{}
\title{Task allocation for multi-agent systems via unequal-dimensional optimal transport}
\author{\thanks{This research has been supported by the Swedish Research Council (VR) under grant 2020-03454, KTH Digital Futures, Swedish Research Council Distinguished Professor Grant 2017-01078, and Knut and Alice Wallenberg Foundation Wallenberg Scholar Grant.} Anqi Dong\thanks{Anqi Dong is  with the Division of Decision and Control Systems and Department of Mathematics, KTH Royal Institute of Technology, SE-100 44 Stockholm, Sweden; \tt anqid@kth.se}, Karl H. Johansson,\thanks{Karl H. Johansson is with the Division of Decision and Control Systems and  Digital Futures, KTH Royal Institute of Technology, SE-100 44 Stockholm, Sweden; \tt kallej@kth.se} and Johan Karlsson  \thanks{Johan Karlsson is  with Department of Mathematics and Digital Futures, KTH Royal Institute of Technology, SE-100 44 Stockholm, Sweden; \tt johan.karlsson@math.kth.se}
}

\maketitle
\thispagestyle{empty}

\begin{abstract}
We consider a probabilistic model for large-scale task allocation problems for multi-agent systems, aiming to determine an optimal deployment strategy that minimizes the overall transport cost. 
Specifically, we assign transportation agents to delivery tasks with given pick-up and drop-off locations, pairing the spatial distribution of transport resources with the joint distribution of task origins and destinations. This aligns with the optimal mass transport framework where the problem and is in the unequal-dimensional setting. The task allocation problem can be thus seen as a linear programming problem that minimizes a quadratic transport cost functional, optimizing the energy of all transport units.  The problem is motivated by time-sensitive medical deliveries using drones, such as emergency equipment and blood transport. In this paper, we establish the existence, uniqueness, and smoothness of the optimal solution, and illustrate its properties through numerical simulations.

\end{abstract}


\section{Introduction}\label{sec:intro}
\IEEEPARstart{E}{fficiently} assigning transportation resources to fulfill delivery requests remains a critical challenge nowadays, especially under scalability and real-time computational constraints. The rise of drone-based delivery has demonstrated their potential for rapid, flexible, and energy-efficient transport of medical supplies, such as blood, defibrillators, first-aid kits, and emergency medications \cite{li2022blood, ling2019aerial, schierbeck2023drone,hii2019evaluation}. Unlike ground-based transport, aerial systems bypass traffic congestion, making them particularly advantageous in urban environments and remote areas where immediate medical intervention is essential. This problem can be framed as a task allocation problem in multi-agent systems, where available transport agents are assigned to fulfill delivery requests. Similar challenges arise in taxi dispatch, food delivery, and mail services, with applications across economics, control \cite{guo2013revising}, and optimization \cite{burkard2012assignment}.

Optimal mass transport (OMT), originating in the 18th century~\cite{monge1781memoire}, is a classical framework for optimizing transportation schedules by balancing supply and demand while minimizing costs. Kantorovich's pioneering work on duality and linear programming~\cite{kantorovich1942translocation}, laid the foundation for significant advancements in OMT. Over the past two decades, contributions by McCann, Villani, and others have extended OMT to diverse fields, including data science \cite{peyre2019computational, Stephanovitch_Dong_Georgiou_2024}, genomics \cite{tronstad2025multistageot}, and economics \cite{chiappori2020multidimensional}. Classical OMT results for existence and uniqueness of solutions, such as the Monge condition, are restricted to the \emph{equal-dimensional} setting, where distributions of the same dimensionality are matched. However, recent extensions to \emph{unequal-dimensional} OMT \cite{chiappori2017multi, chiappori2020multidimensional, mccann2020optimal} have been made for matching distributions of different dimensions. This generalization has proven valuable in economics, resource management, and social sciences. For example, in hedonic pricing, products with multi-attribute features (quality, durability, brand reputation) are optimally paired with consumer preferences. In labor markets, employees with diverse skill sets are assigned to jobs characterized solely by wage.

As the number of transportation agents, delivery tasks, and size of the delivery region grow overwhelmingly large, a microscopic approach that explicitly tracks individual pairwise assignments becomes computationally infeasible. The macroscopic perspective offers a scalable alternative by probabilistically representing the spatial distributions of drones and delivery demands, and considering the collective dynamics of the system. This allows task allocation to be analyzed as an aggregate optimization problem rather than individual assignments. To this end, we describe the spatial distributions of delivery tasks and transport agents by their respective coordinates. We represent the spatial distributions of delivery tasks and transport agents by their coordinates: $\bx^{O}, \bx^{D}$ denote task origins and destinations, with joint distribution of origin-destination pairs as $\mu(\bx^{O}, \bx^{D})$, while $\by$ represents transport agent positions, distributed as $\nu(\by)$. Each agent is ruled to depart from $\by$, picks up a payload at $\bx^{O}$, delivers it to $\bx^{D}$, and returns to $\by$, as shown in Figure \ref{fig:drone}. 
\begin{figure}[htb!]
\centering
\includegraphics[width=0.8\linewidth]{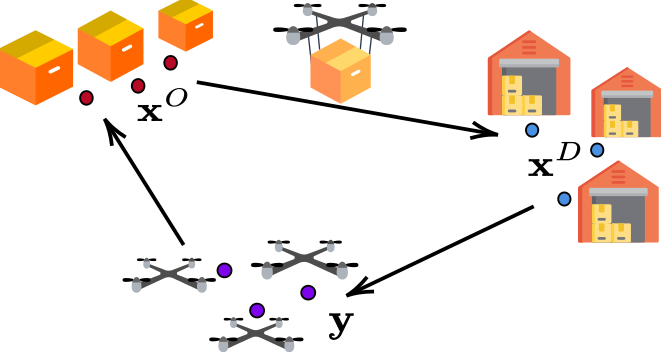}
\caption{Task allocation process: a standby agents begins at $\by$, where the agent picks up shipment at $\bx^O$, delivers them to destination $\bx^D$, and then returns to $\by$.}
\label{fig:drone}
\end{figure}
\noindent This setup exemplifies a mismatch in spatial dimensions -- the standby drones’ locations are represented in $n$ dimensions, whereas the origin-destination (OD) pairs form a $2n$-dimensional space. Consequently, the task allocation problem can be formulated as an unequal-dimensional optimal transport problem, where resources and demand are optimally coupled across distinct spatial dimensions.


In this work, we introduce a novel unequal-dimensional optimal mass transport (OMT) framework for task allocation, addressing the dimensional mismatch between transport agents and delivery requests. We establish the existence and uniqueness of the optimal solution by verifying the twist condition and prove its smoothness by showing that the cost function is in index form, effectively acting as a projection from the high-dimensional origin-destination space to the lower-dimensional agent space. We further validate the framework through numerical simulations, demonstrating its efficiency and the geometric properties of the transport plan.

The rest of the paper is organized as follows. Section \ref{sec:prelim} introduces the necessary background and notations. In Section \ref{sec:formulation}, we formulate the task allocation problem as an unequal-dimensional OMT problem and establish its theoretical foundations, including existence, uniqueness, and smoothness of the solution. Section \ref{sec:example} presents numerical examples. Finally, Section \ref{sec:conclusion}  concludes the article with future directions.

\section{Optimal Transport Preliminaries}\label{sec:prelim} 

The standard optimal matching problem considers two distributions, $\mu$ and $\nu$, which are assumed to be probability measures (i.e., with total mass normalized to 1), and seeks a matching plan that pairs \emph{supply} distribution $\mu$ with \emph{demand} distribution $\nu$. In the Monge formulation \cite{monge1781memoire}, the we seek a optimal transport map $T: x \in \mathbb R^n \mapsto y = T(x)\in \mathbb R^m $  that minimizes the cost functional
\begin{equation*}
\int_{\mathbb{R}^n} c(x, T(x)) \mu(dx),    
\end{equation*}
while aligning $\mu$ with $\nu$, i.e., such that $T_\sharp \mu = \nu$.
Here, $T_\sharp$ denotes the \emph{push-forward} of the distribution, that for any Borel set $S \subset \mathbb{R}^m$, $\mu(T^{-1}(S)) = \nu(S)$.


The Kantorovich relaxation \cite{kantorovich1942translocation} generalizes the Monge problem by considering a coupling measure $\pi$ on the product space $\mathbb{R}^n \times \mathbb{R}^m$. The coupling $\pi$ specifies a probabilistic matching between $\mu$ and $\nu$, with marginals $\mu$ and $\nu$. The objective yields to find $\pi$ that minimizes the cost functional
\begin{equation*}
\iint_{\mathbb{R}^n \times \mathbb{R}^m} c(x, y) \pi(dx, dy),
\end{equation*}
that abiding marginal constraints 
\begin{align*}
\int_{\mathbb{R}^m}  \pi(dx, dy) = \mu(dx) \quad \mbox{and} \quad \int_{\mathbb{R}^n}  \pi(dx, dy) = \nu(dy).
\end{align*}
It is well-known that, under mild conditions, an optimal coupling measure $\pi$ exists. Specifically, if the cost function $c(x, y)$ is lower semi-continuous and the marginal measures $\mu$ and $\nu$ are probability measures on Polish spaces (i.e., complete, separable metric spaces), then there exists a measure $\pi$ that minimizes the cost functional \cite{villani2021topics}. 

Regarding the uniqueness of the coupling measure $\pi$, in the one-dimensional case ($n=m=1$), the cost function $c(x, y): \mathbb{R} \times \mathbb{R} \to \mathbb{R}$ satisfies the \emph{Monge condition} (also, \emph{Spence-Mirrlees condition} in economics) in \cite{rachev1998mass, santambrogio2015optimal}, if the cross difference satisfies
\begin{align}\label{eq:cross_diff}
c(x, y) + c(x^{\prime}, y^{\prime}) - c(x, y^{\prime}) - c(x^{\prime}, y) \leq 0,
\end{align}
for $x^{\prime} \geq x$ and $y^{\prime} \geq y$. Under this condition, the minimizer of the optimal transport problem is unique if and only if the cost function $c(x, y)$ is quasi-antitone/monotone, as in \cite[Theorem 3.1.2]{rachev1998mass}. Intuitively, in matching theory, the cross difference condition ensures that for given matching pairs  $(x, y)$ and $(x^{\prime},y^{\prime})$, reassigning to $(x,y')$ and $(x',y)$ does not improve the total cost. This guarantees that the original matching is both optimal and stable.

In addition, the Monge condition is satisfied for $c(x, y) \in C^1$ if the \emph{twist condition} holds, i.e.,
\begin{equation}\label{eq:twistcond}
D_x c(x, y) \neq D_x c(x, y^{\prime}) \mbox{ and } D_y c(x, y) \neq D_x c(x^{\prime}, y),
\end{equation}
for all $x,y$ and any distinct pairs $y \neq y^{\prime}, x \neq x^{\prime}$, respectively\cite{villani2021topics}. This condition ensures that the mapping $y \mapsto \nabla_x c(x, y)$ is injective for each fixed $x$. In other words, the gradient $\nabla_x c(x, y)$ uniquely determines $y$, and similarly, $\nabla_y c(x, y)$ uniquely determines $x$, ensuring a well-defined and unique pairing/matching plan.

The classical Monge condition may be further refined for $C^2$ cost functions by imposing a stricter requirement on the mixed partial derivative of the cost function $c(x, y)$. Specifically,
\begin{equation*}
\frac{\partial^2 c}{\partial x \partial y} > 0 \quad \text{or} \quad \frac{\partial^2 c}{\partial x \partial y} < 0,
\end{equation*}
that indicates the mapping $y \mapsto \frac{\partial c}{\partial x}(x, y)$ is strictly monotonic, and thus injective, for each fixed $x$. Therefore, the twist condition may be considered as a non-local generalization of the Monge condition.

\begin{definition}[\bf \textit{Pure matching}]
If $T$ is concentrated on a graph of some function $T: \mathbb R^n \mapsto \mathbb R^m$, the stable matching is called pure. This implies that almost all agents of type $x$ must match exclusively with the same agents of type $y = T(x)$.
\end{definition}

In economics, a stable matching refers to an allocation plan  $(x, y) \in \mathbb R^n \times \mathbb R^m$  where there exist payoff functions  $u(x)$  and  $v(y)$  satisfying the budget constraint:
\begin{subequations}
    \begin{equation}
    u(x) + v(y) \leq c(x,y),    
    \end{equation}
    and the reverse inequality
    \begin{equation}
    u(x) + v(y) \geq c(x,y),   
    \end{equation}
\end{subequations}
where $c(x,y)$ is the \emph{transferable utility}. These conditions ensure that no individual can benefit by deviating from the assigned matching, guaranteeing stability in the allocation.

\section{Unequal-dimensional Optimal Transport Formulation for Optimal Task Allocation}\label{sec:formulation}

The classical OMT problem with quadratic cost \cite{benamou2000computational} can be viewed as a special case of a control problem, where the goal is to steer the probability density of a state vector from an initial distribution $\mu(x)$ to a final distribution $\nu(y)$ over a time interval $[0, t_f]$, while minimizing the energy 
\begin{subequations}\label{prob:bb}
\begin{equation}\label{eq:dynamic}
\mathbb E \left\{ \int_0^{t_f} \| u(t) \|^2\right\} \, dt,     
\end{equation}
subject to the system dynamics
\begin{align}\label{eq:count}
&\dot{x} = u, \quad x(0) \sim \mu(x), \mbox{ and } x(t_f) \sim \nu(y).
\end{align}
The system \eqref{eq:count} can be interpreted as a linear dynamical system with trivial prior dynamics $\dot{x} = u(t)$ and the velocity field $v(t,x)$ acts as a control input $u(t)$.
This can be generalized to general dynamics,
$\dot{x} = A(t)x(t) + B(t)u(t)$, 
capturing the collective dynamics of multi-agent systems \cite{chen2016optimal, chen2021optimal, dong2024monge, chen2018state}, with applications in swarm control, flow modeling, and collective motion of particles.
\end{subequations}

We are now in the position to introduce a non-trivial generalization of the equal-dimensional optimal transport framework. This approach models task allocation in multi-agent systems by matching the spatial distribution of agents to the distribution of origin-destination task pairs while minimizing travel costs. We start with the formulation of the problem.

The spatial locations of a task’s origin, destination, and the awaiting agent are represented by their respective coordinates
\begin{equation*}
\bx^{O} =
\begin{bmatrix}
x^O_1 \\
\vdots \\
x^O_n
\end{bmatrix}, \quad
\bx^{D} =
\begin{bmatrix}
x^D_1 \\
\vdots \\
x^D_n 
\end{bmatrix}, \quad \mbox{and }
\by =
\begin{bmatrix}
y_1 \\
\vdots \\
y_n
\end{bmatrix},    
\end{equation*}%
where $\bx^{O}$ represents the origin coordinates of the delivery task, $\bx^{D}$ the destination coordinates, and $\by$ the location of the waiting agent. We define the trip expense using the squared Euclidean norm:
\begin{align*}
c(\bx^O, \bx^D, \by) = \|\bx^O - \by\|_2^2 + \|\bx^O - \bx^D\|_2^2 + \|\bx^D - \by\|_2^2,
\end{align*}
that may be decomposed as following elementary components:
\begin{enumerate}
\item[-] \emph{Pickup cost} $\|\bx^O - \by\|_2^2$: travel from agents’ location $\by$ to task origin $\bx^O$;
\item[-] \emph{Shipping cost} $\|\bx^O - \bx^D\|_2^2$: transport from $\bx^O$ to destination $\bx^D$;
\item[-] \emph{Returning cost} $\|\bx^D - \by\|_2^2$: return from $\bx^D$ to $\by$.
\end{enumerate}    
The spatial distribution of awaiting agents is modeled as $\nu(\by)$ over $\mathbb{R}^{n}$, where $\by$ denotes the agents’ locations.\footnote{Typically, $n=2 \mbox{ or } 3$ in the allocation task.} Similarly, delivery requests are represented by $\mu(\bx^O, \bx^D)$, with $\bx^O$  as the origin and $\bx^D$ as the destination.

The joint probability distribution $\pi(\bx^O, \bx^D, \by)$, known as a coupling, represents the entire allocation process by specifying how agents at locations $\by$ are assigned to a task with origin-destination pairs $(\bx^O, \bx^D)$. This coupling quantifies the proportion of agents dispatched to pick up shipments at origins $\bx^O$ and transport them to destinations $\bx^D$. Within this framework, $\pi(\bx^O, \bx^D, \by)$ captures the spatial dependencies between agents and tasks, ensuring an optimal allocation that minimizes travel cost $\bc(\bx^O, \bx^D, \by)$ for all transport agents.

The problem focuses on identifying an optimal allocation plan with the distributions of delivery tasks' origin and destination provided. We assume that the joint distribution of origins and destinations, denoted as $\mu(\bx^O, \bx^D)$, is known. In particular, let $\nu(\by)$ represent the distribution of awaiting agents and $\mu(\bx^O, \bx^D)$ represent the distribution of shipment requests. The objective is to optimize the allocation plan $\pi(\bx^O, \bx^D, \by) \in \mathbb{R}^n \times \mathbb{R}^n \times \mathbb{R}^n$ that minimizes the total transportation cost defined as 
\begin{align}\label{eq:obj_o}
\langle c, \pi \rangle = \iiint_{\bx^O,\bx^D,\by} c(\bx^O, \bx^D, \by)  d\pi(\bx^O, \bx^D, \by)
\end{align}
subject to the agent availability constraints and the delivery demands
\begin{subequations}
\begin{equation*}
\iint_{\bx^O, \bx^D} \!\!\!\!\!\!\! d\pi(\bx^O, \bx^D, \by) = \nu(\by),
\int_{\by} d\pi(\bx^O, \bx^D, \by) = \mu(\bx^O, \bx^D),    
\end{equation*}  
\end{subequations}
where $\mu(\bx^O, \bx^D)$ conceptualize the portion of demands traveling from $\bx^O$ to $\bx^D$. The marginal condition is also known as \emph{market clearing} criterion in economic literature.

Notably, the problem in \eqref{eq:obj_o} aligns with the unequal-dimensional optimal transport framework \cite{chiappori2020multidimensional}, where the matching occurs between distributions of different dimensions. In order to show this, we first redefine the delivers’ origin-destination pair as
\begin{equation*}
\bx =
\begin{bmatrix}
\bx^1\\
\bx^2
\end{bmatrix}
=
\begin{bmatrix}
\bx^O\\
\bx^D
\end{bmatrix}.
\end{equation*}
The traveling cost can then be reformulated as
\begin{equation}\label{eq:travel-cost}
c(\bx, \by) = \|\by - \bx^1\|_2^2 + \|\bx^1 - \bx^2\|_2^2 + \|\by - \bx^2\|_2^2,    
\end{equation}
where $\by$ is the location of the awaiting transportation resource. With the redefined notations, the problem can be equivalently written as follows.

\vsp
\begin{problem}[\bf \textit{Unequal-dimensional OMT}]\label{prob:1}
Given the distribution of standby agents $\nu(\by)$ over $\calY \subseteq \mathbb{R}^n$ and the distribution of deliveries' origin-destination requests $\mu(\bx)$ over $\calX \subseteq \mathbb{R}^{2n}$. The problem is to seek an allocation plan $\pi(\bx, \by)$ by minimizing
\begin{equation}\label{eq:obj_1}
\langle c, \pi \rangle = \iint_{\bx,\by} c(\bx, \by)  d\pi(\bx, \by)
\end{equation}
subject to the marginal constraint on agent availability
\begin{subequations}
\begin{equation}
\int_{\bx} d\pi(\bx, \by) = \nu(\by),    
\end{equation}
while guarantees that all delivery requests are fulfilled
\begin{equation}
\int_{\by} d\pi(\bx, \by) = \mu(\bx) =\mu(\bx^O, \bx^D).   
\end{equation}   
\end{subequations}
\end{problem}
\vsp 

In the remainder of this section, we will establish the existence, uniqueness, and smoothness of the solution to Problem~\ref{prob:1} within the framework of unequal-dimensional transportation and matching. Specifically, we show that a solution always exists from the lower semi-continuity of the cost function; the optimal matching is unique under an appropriately generalized twist condition \eqref{eq:twistcond}; and the unique matching is also smooth by showing that the cost function is in a so-called \emph{index form}.

\vsp
\begin{proposition}[\bf \textit{Existence}]\label{prop:existence}
Problem \ref{prob:1} admits a solution.
\end{proposition}

\begin{proof}
To establish the existence of a solution, we first observe that $\calX \subset \mathbb{R}^{2n}$ and $\calY \subset \mathbb{R}^n$ are compact subsets. Consequently, their product space $\calX \times \calY$ is also compact, ensuring that any sequence of couplings $\{\pi_n\}$ defined on $\calX \times \calY$ has a subsequence that converges weakly to some $\pi^\ast$.

Next, the cost functional
$
\int_{\calX \times \calY} c(\bx, \by) \, d\pi(\bx, \by)    
$
is lower semi-continuous with respect to the weak convergence of measures. The set of admissible couplings is given by
\begin{align*}
&\Gamma(\mu, \nu) = \bigg\{\pi \in \mathcal{P}(\calX \times \calY) \ \Big| \\
&\phantom{xxxxxxxxxxxx} \int_{\calX} d\pi(\bx, \by) = \nu(\by), \, \int_{\calY} d\pi(\bx,\by) = \mu(\bx) \bigg\},   
\end{align*}
where $\mathcal{P}(\calX \times \calY)$ denotes the space of probability measures on $\calX \times \calY$. The set $\Gamma(\mu, \nu)$ is convex due to the linearity of the marginal constraints.

The existence of an optimal coupling is guaranteed since the cost functional is lower semi-continuous and the admissible set $\Gamma(\mu, \nu)$ is weakly compact.
\end{proof}
\vsp

Next, to show Problem \ref{prob:1} admits a unique minimizer, we validate the $x$-twist condition is satisfies for cost $c(\bx,\by)$ in \eqref{eq:travel-cost}
in unequal-dimensional OMT.

\vsp
\begin{definition}[\bf \textit{x-Twist condition}]\label{def:twist}
The $x$-twist condition is that, for fixed $\by \in \mathbb{R}^m$, the gradient of the cost function with respect to $\bx\in\mathbb R^n$ ($n>m$), $\nabla_\bx c(\bx,\by)$, is injective, i.e.,
\[
\by \mapsto \nabla_\bx c(\bx,\by),
\]
is injective for fixed $\by$.   
\end{definition}

\vsp
\begin{proposition}[\bf \textit{Uniqueness}]\label{prop:unique}
The cost function satisfies the $x$-twist condition, and the optimal allocation plan in Problem~\ref{prob:1} is thus pure and ensures uniqueness.
\end{proposition}

\begin{proof}
The result follow instantly by computing gradient of $c(\bx, \by)$ with respect to $\bx = \begin{bmatrix} (\bx^O)^T, (\bx^D)^T\end{bmatrix}^T$ as 
\begin{equation*}
\nabla_\bx c(\bx, \by) =
\begin{bmatrix}
[\partial c/\partial x^O_i]_{i=1}^n \\
[\partial c/\partial x^D_i]_{i=1}^n
\end{bmatrix}
= 
\begin{bmatrix}
2(\bx^O - \by) + 2(\bx^O - \bx^D) \\
2(\bx^D - \by) + 2(\bx^D - \bx^O) \\
\end{bmatrix},  
\end{equation*}
and the inequality 
\begin{equation*}
\nabla_\bx c(\bx, \by) - \nabla_\bx c(\bx, \by') \neq 0,
\end{equation*}
thus holds for all $\bx$ if and only if $\by\neq \by'$. The injectivity of $\by \to \nabla_\bx c(\bx,\by)$ for each $\bx$ ensures the uniqueness and purity of the match.
\end{proof}
\vsp

\begin{definition}[\bf \textit{Non-degeneracy condition} \cite{chiappori2020multidimensional}]\label{def:non-deg}
The cost $c$ is said to satisfy the non-degeneracy condition 
\begin{align*}
\text{rank}\left(\nabla^2_{\bx\by} c(\bx,\by)\right) = \min\{m, n\},  \quad \mbox{ for } x\in \mathcal{X}, y\in \mathcal{Y}   
\end{align*}
where $n$ and $m$ are the dimensions of $\bx$ and $\by$.   
\end{definition}

\vsp
\begin{proposition}
The cost function $c(\bx, \by)$ in \eqref{eq:travel-cost} satisfies the non-degeneracy condition.
\end{proposition}

\begin{proof}
By definition, the Hessian of $c(\bx,\by)$ with respect to 
$\bx = \begin{bmatrix} (\bx^O)^T, (\bx^D)^T\end{bmatrix}^T$
is 
\begin{align*}
\nabla_{\bx\by}^2 c(\bx,\by) = 
\begin{bmatrix}
\left(\frac{\partial^2 c}{\partial x^O_{i} \partial \by_{j}}  \right)_{{i, j=1}}^{n}\\[0.14in]
\left(\frac{\partial^2 c}{\partial x^D_{i} \partial \by_{j}}  \right)_{{i, j=1}}^{n}
\end{bmatrix}=
\begin{bmatrix}
-2~\mathbf I_{n\times n}\\
-2~\mathbf I_{n\times n}
\end{bmatrix},
\end{align*}
where $\mathbf I$ is the $n$-by-$n$ identity matrix and the result follows since the matrix is full rank for all $x$ and $y$.
\end{proof}
\vsp

To obtain the potential indifference set $X(\by, \bk)$ in \cite[Section 2.2.1]{chiappori2020multidimensional} that defined as  
\begin{align*}
X(\by, \bk) := \left\{\bx \in \calX ~\,\big|\,~ \nabla_\by c(\bx,\by) = \bk\right\},   
\end{align*}
with given $\bk$, we combine the partial derivatives into the gradient
\begin{align*}
\nabla_{\mathbf{y}} c(\mathbf{x}, \mathbf{y}) =
\begin{bmatrix}
4\by - 2(\bx^O + \bx^D) \\
\end{bmatrix}
\end{align*}
The indifference set $X(\by, \bk)$ is thus defined as 
\begin{equation*}
\left \{\bx \in \calX \; \big |\, \;
\begin{aligned}
4\by - 2(\bx^O + \bx^D) = \bk
\end{aligned}\right\}.
\end{equation*}

We are now ready to discuss the property of \emph{smoothness} -- if two agents in $\calX$, say $\mathbf{x}$ and $\mathbf{x}^{\prime}$, have ``similar characteristics'' (i.e., are close in type), their corresponding matches $T(\mathbf{x})$ and $T(\mathbf{x}^{\prime})$ in $\calY$ (the paired or matched agents) should also have similar characteristics. It guarantees that small variations in agent attributes do not lead to disproportionately large changes in their assigned matches, promoting stability and robustness in the allocation process.

\vsp
\begin{remark}[\bf \textit{Smoothness}]
The mapping $T: \mathbb R^n\to \mathbb R^m$ is said to be smooth \cite{chiappori2020multidimensional} if the process of generating the matching is continuous. Qualitatively, this addresses whether tasks’ origin-destination pairs $\mathbf{x}$ and $\mathbf{x}^{\prime}$, which are close in type, are matched to agents with similarly spatial characteristics.
\end{remark}
\vsp

The \emph{nestedness condition} for determining the smoothness property of the function, by defining the 
\begin{equation*}
X_{\leq}(\by,\bk) = \left\{\bx \in \calX \ | \ \nabla_\by c(\bx,\by)  \leq \bk\right\},   
\end{equation*}
where $\bk$ is chosen so that
\begin{equation*}
\mu(X_{\leq}(\by,\bk)) = \nu\big((-\infty, \by)\big).
\end{equation*}
The model is said to be nested if, whenever $\by < \by'$, the sub-level sets are nested so that 
\begin{equation}\label{eq:nestness}
X_{\leq}(\by, \bk(\by)) \subset X_{<}\left(\by', \bk(\by')\right),    
\end{equation}
where $X_{<}(\by,\bk) = \Big\{x \in \calX \ | \ \nabla_\by c(\bx,\by) < \bk\Big \}$.

The nestedness condition \eqref{eq:nestness} is known to be the sufficient condition for a unique smooth assignment plan, according to its definition, the general nestedness condition, compared with other conditions for determining a unique solution in equal dimensional transport, depends not only on the cost itself, but the triplet $\big\{c(\bx,\by),\mu,\nu\big\}$ as it defines. 

\vsp
\begin{remark}[\bf \textit{Twist \& Nestedness}]
Twist condition (Definition \ref{def:twist}) and nestedness condition are both sufficient to guarantee solution's uniqueness. However, the nestedness condition \eqref{eq:nestness} proposes a constructive way to build a level-set solution, hence converting an unequal-dimensional matching problem to an equal-dimensional one. 
\end{remark}
\vsp

\begin{definition}[\bf \textit{Index and pseudo-index form} \cite{chiappori2017multi}]
The index and pseudo-index form is a special case of the nestedness condition, wherein no additional assumptions are needed on $\mu$ and $\nu$ except for absolute continuity. A cost function $c(\bx,\by)$ is said to be in index form\footnote{We denote the index form is no more than a projection of the space $\calX$ onto $\calY$ through $\sigma$. Consequently, elements $x \in \mathcal{X}$ from the potential indifference set $X(\mathbf{y}, \mathbf{k})$ lose their unique order and are assigned to the same $\mathbf{y} \in \mathcal{Y}$, exhibiting a ``loss of strict ordering'' within the indifference set $X(\mathbf{y}, \mathbf{k})$. Another example of a cost function in index form is $c(\mathbf{x}, \mathbf{y}) = \frac{\mathbf{x}^2}{\mathbf{y}}$.} if it can be expressed as:
\begin{align}\label{eq:index-form}
c(\bx,\by) = \sigma\big(I(\bx), \by\big),    
\end{align}
where $I: \calX \to \mathbb{R}^n$ is an index function that maps $\bx$ to a lower-dimensional space, and $\sigma: \calX \times \mathbb{R}^m \to \mathbb{R}$ is the monotonic function that depends on $\bx$ and the index  $I(\by)$.

In addition, when $m=1$, a more relaxed, pseudo-index form of the cost function $c(\bx, \by)$ can be defined as
$$
c(\bx,\by) = \alpha(\by) + \sigma\big(I(\bx), \by\big),    
$$
where addition term $\alpha: \calX \to \mathbb{R}$ is a function that depends only on $\bx$.
\end{definition}
\vsp

\begin{proposition}[\bf \textit{Smooth plan}]\label{prop:smooth}
The cost $c(\bx, \by)$ is in index form -- nestedness condition holds regardless of the choice of absolute continuous $\mu$  and  $\nu$. Thus, the optimal allocation plan is unique and also smooth.
\end{proposition}
\vsp

\begin{proof}
To demonstrate that the cost function  $c(\mathbf{x}, \mathbf{y})$ in \eqref{eq:travel-cost} is in index form, we begin by expressing it as 
{\footnotesize
\begin{align*}
c(\mathbf{x}, \mathbf{y}) = 2\big( |\mathbf{y}|^2 - (\mathbf{x}^1)^\top \mathbf{y} + |\mathbf{x}^1|^2 - (\mathbf{x}^1)^\top \mathbf{x}^2 + |\mathbf{x}^2|^2 - (\mathbf{x}^2)^\top \mathbf{y} \big).    
\end{align*}}%
This factorization allows us to reformulate the original optimization problem
\begin{subequations}
\begin{equation}\label{eq:obj_org}
\min_{\pi} \iint_{\bx,\by} c(\bx, \by) \ d\pi(\bx, \by)
\end{equation}
as an equivalent problem minimizing the negative correlations
\begin{equation}\label{eq:obj_inter}
\min_{\pi} \iint_{\bx,\by} -\left(\bx^1 + \bx^2\right)^T\!\! \by\  d\pi(\bx, \by)
\end{equation}    
\end{subequations}
In fact, $\pi$ is a minimizer for \eqref{eq:obj_inter} if and only if it minimizes \eqref{eq:obj_org}, see also \cite{villani2021topics}.
This term in the objective can be represented in index form \eqref{eq:index-form} as
\begin{equation*}
\hat c(\mathbf{x}, \mathbf{y}) =-(\bx^1+\bx^2)^T\by= \sigma\big(I(\mathbf{x}), \mathbf{y}\big),   
\end{equation*}
where the projection operator  $I: \mathcal{X} \to \mathbb{R}^n$  is defined by
\begin{equation*}
I(\mathbf{x}) = -(\mathbf{x}^1 + \mathbf{x}^2),    
\end{equation*}
and the bilinear function  $\sigma: \mathbb{R}^n \times \mathbb{R}^n \to \mathbb{R}$  is given by
\begin{equation*}
\sigma(\mathbf{u}, \mathbf{y}) = -\mathbf{u}^\top \mathbf{y}. 
\end{equation*}
The uniqueness and smoothness of the matching follow directly from \cite{chiappori2020multidimensional}.
\end{proof}
\vsp

\begin{remark}
The same result in Proposition \ref{prop:smooth} does not hold for Euclidean norm
$
c(\bx, \by) = |\by - \bx^1|_2 + |\bx^1 - \bx^2|_2 + |\by - \bx^2|_2.
$    
A counterexample follows directly from an extension of the well-known \emph{book shifting} example.
\end{remark}
\vsp

\begin{remark}
The results established in Proposition \ref{prop:existence}, \ref{prop:unique}, and \ref{prop:smooth} can be generalized from problem in \eqref{prob:bb} to a controllable linear system with dynamics
$$
\dot{x}(t) = A(t)x(t) + B(t)u(t),
$$
known as OMT with prior dynamics (OMT-wpd) \cite{chen2016optimal}. Given the state transition matrix $\Phi(t, s)$, the controllability Gramian reads
\begin{align*}
M(t, s) = \int_s^t \Phi(t, \tau) B(\tau) B(\tau)^T \Phi(t, \tau)^T d\tau,  
\end{align*}
that is positive definite under controllability assumption. The optimal cost of transporting mass between states $x$ and $y$ over the time interval $[0,1]$ is given by
$
c(x, y) = \frac{1}{2} (y - \Phi_{10} x)^T M_{10}^{-1} (y - \Phi_{10} x),    
$
Rewriting OMT-wpd in the standard Kantorovich formulation is achieved by introducing the linear transformation
\begin{align*}
C: \begin{bmatrix} 
x \\ y 
\end{bmatrix} 
\to 
\begin{bmatrix} \hat{x} 
\\ \hat{y} 
\end{bmatrix} 
= 
\begin{bmatrix} 
M_{10}^{-1/2} \Phi_{10} x 
\\ M_{10}^{-1/2} y 
\end{bmatrix},
\end{align*}
leading to the reformulation
$
\min_{\hat{\pi}} \int_{\mathbb{R}^n \times \mathbb{R}^n} \frac{1}{2} \| \hat{y} - \hat{x} \|^2 \hat{\pi}(d\hat{x} d\hat{y}).
$
We refer interested readers to \cite[Section 3]{chen2016optimal} for a more comprehensive and detailed investigation of OMT-wpd.
\end{remark}
\vsp

\section{Illustrative Example}\label{sec:example}
We present illustrative examples in this section to demonstrate our proposed framework. The code for all experiments conducted below is available at {\blue \url{github.com/dytroshut/task-allocation}}. The discretized version of Problem \ref{prob:1} can be formulated as a linear programming problem, which adopted a rich library of numerical solvers. For all numerical examples, we utilized CVX \cite{cvx} as the optimization solver.

We first present one-dimensional example and visualize the result. Specifically, we assume $\mathbf{x}^1, \mathbf{x}^2, \mathbf{y} \in \mathbb{R}$, allowing the matching $\mathbf{y} = T(\mathbf{x}): \mathbb{R}^2 \to \mathbb{R}$ to be plotted in 3-dimenional space. The origin-destination requests from customers are represented by the measure $\mu(\mathbf{x}^1, \mathbf{x}^2)$, which aims to be matched with the available agents described by the measure $\nu(\mathbf{y})$. The smoothness of matching is shown in Fig. \ref{fig:fixed1}.
\begin{figure}[htb!]
    \centering
    \includegraphics[trim={0.5cm 0.2cm 0.2cm 1.1cm},clip,width=0.75\linewidth]{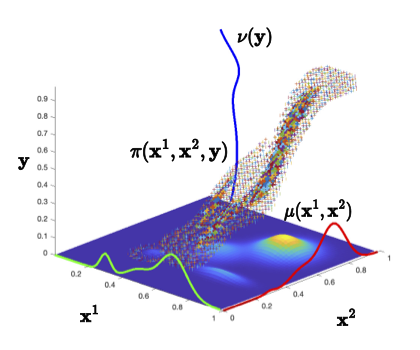}
    \caption{Two- to one-dimensional matching with fixed drone spatial distribution $\nu(\by)$ and the OD distribution $\mu(\bx^1,\bx^2)$ given. The size of the crosses corresponds to the mass at each point of $\pi(\mathbf{x}^1, \mathbf{x}^2, \mathbf{y})$.}
    \label{fig:fixed1}
\end{figure}

Next, we consider a synthetic example in Fig. \ref{fig:example2}, each drone’s spatial position is represented as a three-dimensional vector $\by$, and the origin-destination request is described by a six-dimensional vector $\bx$. We begin with three task with origin-destination pair $\left(\bx^{O}(1), \bx^{D}(1)\right)$, $\left(\bx^{O}(2), \bx^{D}(2)\right)$, and $\left(\bx^{O}(3)), \bx^{D}(3)\right)$. Three available drones providers, located at $\by_1$, $\by_2$, and $\by_3$ are waiting to be matched with three tasks.
\begin{figure}[htb!]
    \centering
    \includegraphics[trim={0.85cm 0.2cm 0.2cm 0cm},clip,width=0.493\linewidth]{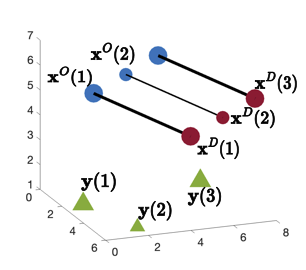}
    \includegraphics[trim={0.85cm 0.2cm 0.2cm 0cm},clip,width=0.493\linewidth]{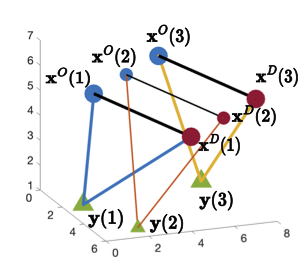}
    \caption{The drones positioned at three distinct locations are matched to three tasks with their origin-destination (OD) given.}
    \label{fig:example2}
\end{figure}

Lastly, we present an example of a medical delivery task in the city of Stockholm, Sweden. Assuming there are 30 ride requests randomly generated at specific locations in the city, the goal is to match drones distributed throughout the area to fulfill these requests. For simplicity, we assume that all drones operate at a fixed altitude, meaning $\bx^{O}_3 = \bx^{D}_3 = \text{constant}$, which is omitted from the formulation. The origin and destination of each ride are marked in red and connected by a straight line to indicate the trip as shown in Fig. \ref{fig:example3}.

\begin{figure}[htb!]
    \centering %
    \includegraphics[trim={1.72cm 1cm 0 0},clip,width=0.8\linewidth]{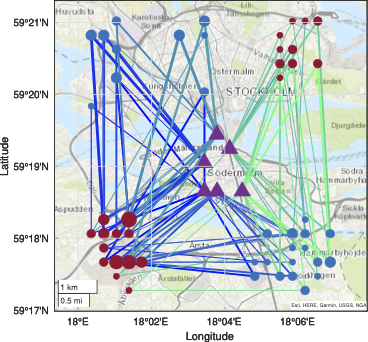}
    \caption{Thirty allocation tasks, with origins in blue and destinations in red, are randomly generated within specific locations in Stockholm, Sweden. Drones, represented by purple triangles, are distributed throughout the city and matched to fulfill these requests.}
    \label{fig:example3}
\end{figure}

\section{Conclusion}\label{sec:conclusion}
In this paper, motivated by medical delivery by drone, we propose a task allocation problem for macroscopic multi-agent system as unequal-dimensional optimal transport problem, modeling drone locations and origin-destination requests as probability distributions. With a cost functional that reflects allocation expenses, we established the existence, uniqueness, and smoothness of the optimal allocation plan. Numerical experiments validated the framework, demonstrating its effectiveness in cost-efficient task allocation. Future work will address unbalanced supply-demand scenarios and shared rides with intermediate stops. We also aim to develop distributed algorithms for large-scale problems and integrate real-time data for dynamic adaptation.


\section*{References}
\bibliographystyle{plain}
\bibliography{references}

\end{document}